\documentclass[11pt]{article}

\PassOptionsToPackage{numbers, compress}{natbib}
\usepackage[preprint]{neurips_2025}

\usepackage[utf8]{inputenc}
\usepackage[T1]{fontenc}
\usepackage{amsmath,amssymb,amsthm}
\usepackage{mathtools}
\usepackage{hyperref}
\usepackage{cleveref}
\usepackage{booktabs}
\usepackage{enumitem}
\usepackage{xcolor}
\usepackage{graphicx}
\usepackage{tikz}
\usetikzlibrary{patterns,decorations.pathreplacing,calc}

\newtheorem{theorem}{Theorem}[section]
\newtheorem{lemma}[theorem]{Lemma}

\newtheorem{proposition}[theorem]{Proposition}
\newtheorem{definition}[theorem]{Definition}
\newtheorem{remark}[theorem]{Remark}

\newtheorem{counterexample}[theorem]{Counterexample}

\newcommand{\R}{\mathbb{R}}

\newcommand{\cl}[1]{\overline{#1}}
\newcommand{\Fix}{\operatorname{Fix}}
\newcommand{\AD}{\operatorname{AD}}
\newcommand{\id}{\operatorname{id}}

\newcommand{\dist}{\operatorname{dist}}

\title{The Defense Trilemma: Why Prompt Injection Defense Wrappers Fail?}

\author{
  Manish Bhatt\thanks{Corresponding author: manish.bhatt13212@gmail.com}\hspace{0.3em}\footnotemark[2] \\
  OWASP, Amazon Leo, UNO \\
  \And
  Sarthak Munshi\footnotemark[2] \\
  Amazon Web Services \\
  \And
  Vineeth Sai Narajala\footnotemark[2] \\
  Cisco \\
  \And
  Idan Habler\footnotemark[2] \\
  Cisco \\
  \AND
  Ammar Al-Kahfah\footnotemark[2] \\
  Amazon Web Services \\
  \And
  Ken Huang \\
  Distributedapps.ai \\
  \And
  Joel Webb \\
  UNO \\
  \And
  Blake Gatto \\
  Shrewd Security \\
  \And
  Md Tamjidul Hoque \\
  UNO
}

\date{}

\begin{document}

\maketitle
\renewcommand{\thefootnote}{\fnsymbol{footnote}}
 \footnotetext[2]{Equal contribution. This work was conducted independently and does not reflect the views, policies, or endorsements of the
  authors' respective employers.}
\begin{abstract}
We prove that no continuous, utility-preserving wrapper defense---a
function $D\colon X\to X$ that preprocesses inputs before the model
sees them---can make all outputs strictly safe for a language model
with connected prompt space, and we characterize exactly where every
such defense must fail. We establish three results under
successively stronger hypotheses: \emph{boundary fixation}---the defense must leave some
threshold-level inputs unchanged; an \emph{$\varepsilon$-robust
constraint}---under Lipschitz regularity, a positive-measure band
around fixed boundary points remains near-threshold; and a
\emph{persistent unsafe region}---under a transversality condition, a
positive-measure subset of inputs remains strictly unsafe. These
constitute a \emph{defense trilemma}: continuity, utility preservation,
and completeness cannot coexist. We prove parallel discrete results
requiring no topology, and extend to multi-turn interactions,
stochastic defenses. The results do not
preclude training-time alignment, architectural changes, or defenses
that sacrifice utility. The full theory is mechanically verified in
Lean~4 with Mathlib
\url{https://github.com/mbhatt1/stuff/tree/main/ManifoldProofs}
(46 files, ${\sim}360$ theorems, no admitted proofs, three standard
axioms) and validated empirically on three LLMs~\cite{munshi2026manifold}.
\end{abstract}

\section{Introduction}
\label{sec:intro}

Can you build a wrapper around a language model that eliminates all
prompt injection vulnerabilities? Most current defense work implicitly
assumes yes. Input classifiers that flag suspicious
prompts~\cite{alon2023detecting}, constitutional rewriting
pipelines~\cite{bai2022constitutional}, and input sanitization
layers~\cite{inan2023llama} all share the same structure: a
function $D\colon X \to X$ that preprocesses prompts before the model
sees them, mapping unsafe inputs to safe equivalents while leaving safe
inputs unchanged. We prove the answer is no, under two constraints. If the defense
is \emph{continuous} (similar prompts produce similar rewrites) and
\emph{utility-preserving} (safe prompts pass through unchanged), it
cannot be \emph{complete} (make every output safe). These three
properties form a \textbf{defense trilemma}: any two can coexist, but
not all three (\Cref{fig:trilemma}).

The impossibility is not about specific attacks or clever prompt
engineering. It arises from the \emph{geometry} of the prompt space
itself: in a connected space, the safe region is open but not closed,
so any continuous defense that fixes safe inputs must also fix points
on the safety boundary. Under successively stronger hypotheses,
we establish three results with progressively stronger conclusions
(\Cref{fig:escalation}):

\paragraph{Boundary fixation (\Cref{thm:main}).}
The defense must fix at least one boundary point, i.e., a prompt where
alignment deviation equals the threshold exactly---passing it through
with no remediation.

\paragraph{$\varepsilon$-robust constraint (\Cref{thm:eps-robust}).}
Under Lipschitz regularity, the defense cannot uniformly reduce
alignment deviation far below~$\tau$ near the fixed boundary point.
For any $x$ within distance $\delta$ of the fixed point~$z$:
\begin{equation}
  f(D(x)) \;\geq\; \tau - LK\,\delta.
\end{equation}

\paragraph{Persistent unsafe region (\Cref{thm:persistent}).}
Under a transversality condition, the alignment surface rises faster
than the defense can pull it down, leaving a positive-measure region
that remains unsafe:
\begin{equation}
  f(D(x)) > \tau \quad\text{for all } x \in \mathcal{S},
  \qquad \mu(\mathcal{S}) > 0.
\end{equation}

\paragraph{From discrete to continuous.}
All three results apply to continuous interpolants of discrete data.
The Tietze extension theorem guarantees that any finite set of
behavioral observations in a normal space admits continuous extensions,
and the impossibilities hold for every such extension
(\Cref{thm:tietze}).

\medskip
\noindent\textbf{Scope and limitations.}
Our results apply specifically to \emph{continuous, utility-preserving
wrapper defenses} on \emph{connected} prompt spaces. They do \emph{not}
preclude effective safety through other mechanisms, including:
\begin{itemize}[nosep]
  \item training-time alignment (RLHF, DPO, constitutional AI training),
  \item architectural changes to the model itself,
  \item discontinuous defenses (e.g., hard blocklists or discrete classifiers),
  \item output-side filters, ensemble defenses, or human-in-the-loop review,
  \item adaptive-threshold systems,
  \item multi-component systems whose classifiers may reject or redirect
        inputs rather than preserving utility on every prompt.
\end{itemize}
In short, our theorems cover only a single continuous wrapper
$D\colon X \to X$ that preprocesses inputs; any mechanism outside this
class is not constrained by our results.

All three conditions---continuity, utility preservation, and
connectedness---are individually necessary; we give counterexamples
for each (\Cref{app:counterexamples}).

\paragraph{Contributions.}
\begin{enumerate}[nosep]
  \item \textbf{Boundary fixation} (\Cref{thm:main}): any continuous,
    utility-preserving defense on a connected Hausdorff space must fix
    a point $z$ with $f(z) = \tau$. Relaxed to score-preserving and
    $\varepsilon$-approximate variants
    (\Cref{thm:score-preserving,thm:eps-relaxed}).
  \item \textbf{$\varepsilon$-robust constraint} (\Cref{thm:eps-robust}):
    under Lipschitz regularity, $f(D(x)) \geq \tau - LK\dist(x,z)$
    for all~$x$. A positive-measure band near~$z$ is constrained
    (\Cref{thm:band-measure}).
  \item \textbf{Persistent unsafe region} (\Cref{thm:persistent}):
    under a transversality condition ($G > \ell(K+1)$, where $\ell$ is
    the defense-path Lipschitz constant), a positive-measure set
    remains strictly above~$\tau$ after defense.
  \item \textbf{Quantitative bounds}: explicit volume lower bound
    (\Cref{thm:coarea}), cone measure bound (\Cref{thm:cone}), and
    an asymmetric defense dilemma (\Cref{thm:dilemma}).
  \item \textbf{Discrete defense dilemma} (\Cref{thm:disc-dilemma}):
    on finite sets, completeness forces non-injectivity (information
    loss); injectivity forces incompleteness.
  \item \textbf{Extensions}: multi-turn (\Cref{thm:multi-turn}),
    stochastic (\Cref{thm:stochastic}), and pipeline
    (\Cref{thm:pipeline}) settings.
  \item \textbf{Lean~4 formalization}: 46 files, ${\sim}360$ theorems,
    zero \texttt{sorry} statements, three standard axioms. Full proofs
    in \Cref{app:proofs}; artifact in \Cref{app:artifact}.
\end{enumerate}

\section{Related Work}
\label{sec:related}

Adversarial robustness research has established that small
perturbations can fool classifiers~\cite{szegedy2013intriguing,
goodfellow2014explaining,carlini2017evaluating,madry2018towards} and
that robustness may fundamentally trade off against
accuracy~\cite{tsipras2018robustness,fawzi2018adversarial}. Certified
defenses provide guarantees for fixed
models~\cite{cohen2019certified,katz2017reluplex,singh2019abstract,
huang2017safety,bagnall2019certifying}, and topological perspectives
have illuminated the structure of decision
boundaries~\cite{naitzat2020topology}. For LLMs specifically,
jailbreaking attacks~\cite{zou2023universal,chao2024jailbreaking,
mehrotra2024tree,greshake2023indirect} and automated
red-teaming~\cite{mouret2015illuminating,samvelyan2024rainbow}
demonstrate persistent vulnerabilities.

Our work differs in kind: rather than studying how failures arise for
fixed models or how specific systems can be certified, we impose
\emph{universal constraints on the defense map itself}. The closest
conceptual precedent is the no-free-lunch
framework~\cite{wolpert1997no}, which proves that no optimizer
dominates across all problems. We prove the analogous result for
wrapper defenses: under continuity and utility preservation, no
defense eliminates all failures.

\Cref{tab:comparison} situates our result among comparable
impossibility and tradeoff theorems.

\begin{table}[t]
\centering
\small
\caption{Positioning among impossibility results (top) and empirical
findings independently consistent with our predictions (bottom).}
\label{tab:comparison}
\begin{tabular}{@{}p{3.4cm}p{2.0cm}p{2.4cm}p{2.4cm}p{2.2cm}@{}}
\toprule
\textbf{Result} & \textbf{Target} & \textbf{Method} & \textbf{Impossibility} & \textbf{Verified} \\
\midrule
Tsipras et al.\ \cite{tsipras2018robustness}
  & Classifier
  & PAC / Gaussian
  & Robustness + accuracy
  & No \\[2pt]
Fawzi et al.\ \cite{fawzi2018adversarial}
  & Classifier
  & Lipschitz geom.
  & Adv.\ invulnerability
  & No \\[2pt]
Wolpert \& Macready \cite{wolpert1997no}
  & Optimizer
  & Counting
  & Universal dominance
  & No \\[2pt]
Cohen et al.\ \cite{cohen2019certified}
  & Smoothed clf.
  & Rand.\ smoothing
  & (Certifies bound)
  & No \\[2pt]
\textbf{This work}
  & \textbf{Def.\ wrapper}
  & \textbf{Topology}
  & \textbf{Cont.\ + util.\ + compl.}
  & \textbf{Lean 4} \\
\midrule
\multicolumn{5}{@{}l}{\textit{Empirical findings consistent with our derived predictions}} \\
\midrule
\textbf{Our prediction} & \textbf{Theorem} & \multicolumn{2}{p{4.8cm}}{\textbf{Confirming evidence}} & \textbf{Pred.} \\
\midrule
Diminishing safety returns
  & Int.\ Stability
  & \multicolumn{2}{p{4.8cm}}{MART~\cite{ge2024mart}; Sleeper Agents~\cite{hubinger2024sleeper}}
  & \S\ref{sec:engineering} \\[2pt]
Long context $\to$ exp.\ harder defense
  & Cost Asymm.
  & \multicolumn{2}{p{4.8cm}}{Many-Shot Jailbreaking~\cite{anil2024many}; Kim et al.~\cite{kim2025manyshot}}
  & \S\ref{sec:engineering} \\[2pt]
Agentic tool-use degrades exp.
  & Pipeline Lip.
  & \multicolumn{2}{p{4.8cm}}{InjecAgent~\cite{zhan2024injecagent}; ASB~\cite{zhang2024asb}}
  & \S\ref{sec:engineering} \\[2pt]
Temperature boundary instab.
  & Stoch.\ Imposs.
  & \multicolumn{2}{p{4.8cm}}{Instability of Safety~\cite{yuan2025instability}}
  & \S\ref{sec:engineering} \\[2pt]
Quantization preserves deep vulns
  & Int.\ Stability
  & \multicolumn{2}{p{4.8cm}}{Skoltech~\cite{skoltech2025quant}; ETH~\cite{eth2025gguf}}
  & \S\ref{sec:engineering} \\[2pt]
Merging preserves weaker parent
  & Int.\ Stability
  & \multicolumn{2}{p{4.8cm}}{Hammoud et al.~\cite{hammoud2024merging}}
  & \S\ref{sec:engineering} \\[2pt]
Adv.\ training fragments
  & Fragment Size
  & \multicolumn{2}{p{4.8cm}}{Whack-a-Mole~\cite{liu2026whackamole}; IRIS~\cite{iris2025}}
  & \S\ref{sec:engineering} \\[2pt]
Asymmetric transfer
  & Transferability
  & \multicolumn{2}{p{4.8cm}}{Weak-to-Strong~\cite{zhao2025weak}; IRIS~\cite{iris2025}}
  & \S\ref{sec:engineering} \\[2pt]
Alignment tax $\propto$ unsafe vol.
  & Discrete Dil.
  & \multicolumn{2}{p{4.8cm}}{Safety Tax~\cite{huang2025safetytax}; Huang et al.~\cite{huang2026formal}}
  & \S\ref{sec:engineering} \\[2pt]
Patch blast radius
  & Patching Nonl.
  & \multicolumn{2}{p{4.8cm}}{Slingshot~\cite{slingshot2026}; Sleeper Agents~\cite{hubinger2024sleeper}}
  & \S\ref{sec:engineering} \\
\bottomrule
\end{tabular}
\end{table}

\section{Formal Framework}
\label{sec:setup}

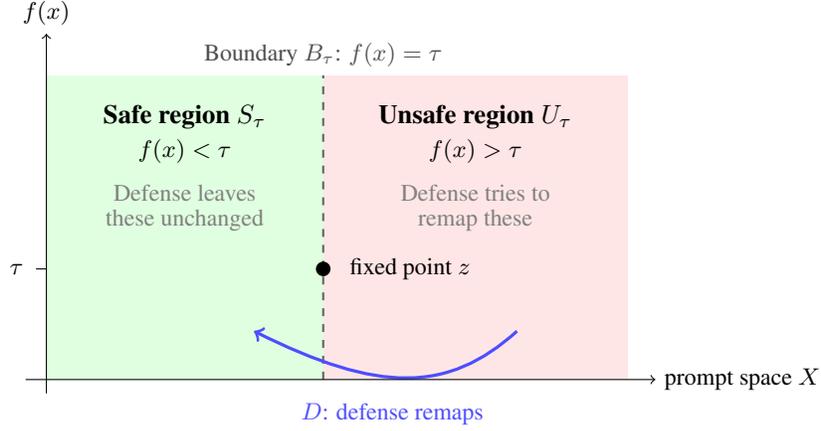
\begin{figure}[t]
\centering
\begin{tikzpicture}[scale=0.92]
  \fill[green!12] (0,0.8) rectangle (4.0,5.2);
  \fill[red!10] (4.0,0.8) rectangle (8.4,5.2);
  \draw[dashed, thick, black!60] (4.0,0.8) -- (4.0,5.2);
  \node[font=\bfseries] at (2.0,4.6) {Safe region $S_\tau$};
  \node[font=\small] at (2.0,4.1) {$f(x) < \tau$};
  \node[font=\small, black!50] at (2.0,3.5) {Defense leaves};
  \node[font=\small, black!50] at (2.0,3.1) {these unchanged};
  \node[font=\bfseries] at (6.2,4.6) {Unsafe region $U_\tau$};
  \node[font=\small] at (6.2,4.1) {$f(x) > \tau$};
  \node[font=\small, black!50] at (6.2,3.5) {Defense tries to};
  \node[font=\small, black!50] at (6.2,3.1) {remap these};
  \node[font=\small, black!70] at (4.0,5.5) {Boundary $B_\tau$: $f(x)=\tau$};
  \fill[black] (4.0,2.4) circle (3pt);
  \node[anchor=west, font=\small] at (4.25,2.4) {fixed point $z$};
  \draw[->, thick, blue!70, line width=1.2pt]
    (6.8,1.5) .. controls (5.8,0.6) and (4.8,0.6) .. (3.0,1.5);
  \node[font=\small, blue!70] at (5.0,0.3) {$D$: defense remaps};
  \draw[->] (-0.3,0.8) -- (8.8,0.8) node[right] {\small prompt space $X$};
  \draw[->] (0,0.6) -- (0,5.8) node[above] {\small $f(x)$};
  \draw (0,2.4) -- (-0.15,2.4);
  \node[anchor=east, font=\small] at (-0.2,2.4) {$\tau$};
\end{tikzpicture}
\caption{Schematic of the prompt space. The defense~$D$ must leave all
safe inputs unchanged (utility preservation) and tries to remap unsafe
inputs into the safe region. Boundary fixation (\Cref{thm:main}) proves
the defense must also leave some boundary points unchanged, the black
dot~$z$ passes through with no remediation.}
\label{fig:landscape}
\end{figure}

\begin{definition}[Alignment Deviation Function]
\label{def:ad}
Let $X$ be a topological space (a set equipped with a notion of
``open'' and ``closed'' sets, such that the whole space and the empty
set are both open and closed, arbitrary unions and finite intersections
of open sets are open, and a notion of ``neighborhoods'' of points---that
is, points ``close to'' a given point). An \emph{alignment deviation function} is
a continuous map $f\colon X \to \R$, i.e., small changes in the
input produce small changes in the score. Given a threshold
$\tau\in\R$:
\begin{align}
  S_\tau &= \{x \in X : f(x) < \tau\}  &&\text{(safe region)} \\
  U_\tau &= \{x \in X : f(x) > \tau\}  &&\text{(unsafe region)} \\
  B_\tau &= \{x \in X : f(x) = \tau\}  &&\text{(boundary)}
\end{align}
\end{definition}

\begin{definition}[Defense]
\label{def:defense}
A \emph{defense} is a continuous map $D\colon X \to X$. It is
\emph{utility-preserving} if $D(x) = x$ for all $x \in S_\tau$
(safe prompts pass through unchanged),
and \emph{complete} if $f(D(x)) < \tau$ for all $x\in X$
(every output is safe).
\end{definition}

The central question is whether a defense can be both
utility-preserving and complete. The following sections show that under
natural conditions the answer is no.

\section{Boundary Fixation}
\label{sec:boundary}

We begin with the most fundamental result. The argument fits in a
paragraph: the defense fixes every safe input, so continuity makes its
fixed-point set closed. But the safe region~$S_\tau$ is open (preimage
of an open interval under continuous~$f$), and in a connected space a
nonempty proper open set is not closed. Hence the fixed-point set
cannot stop exactly at the edge of the safe region, it must spill onto
the boundary~$B_\tau$. Some boundary prompts pass through unchanged.

\begin{theorem}[Boundary Fixation]
\label{thm:main}
Let $X$ be a connected Hausdorff space (a space that is ``in one
piece'' and where distinct points can be separated by neighborhoods).
Let $f\colon X\to\R$ be continuous with $S_\tau, U_\tau \neq
\emptyset$, and $D\colon X \to X$ continuous with
$D|_{S_\tau} = \id$. Then there exists $z\in X$ with $f(z) = \tau$
and $D(z) = z$. Moreover, every $z \in \cl{S_\tau} \setminus S_\tau$
satisfies $f(z) = \tau$ and $D(z) = z$, and this set is nonempty.
\end{theorem}

\begin{proof}[Proof sketch]
In a Hausdorff space, $\Fix(D)$ is closed (preimage of the diagonal).
By utility preservation, $S_\tau \subseteq \Fix(D)$, so
$\cl{S_\tau} \subseteq \Fix(D)$. But $S_\tau = f^{-1}((-\infty,\tau))$
is open and not closed (connectedness: a nonempty proper clopen set
would disconnect~$X$). Hence $\cl{S_\tau} \supsetneq S_\tau$. Any
$z \in \cl{S_\tau}\setminus S_\tau$ satisfies $f(z) = \tau$ (limits of
values ${<}\tau$ cannot exceed~$\tau$, and $z \notin S_\tau$ forces
$f(z) \geq \tau$) and $D(z) = z$. Full proof in \Cref{app:proofs}.
\end{proof}

This means that the defense's fixed-point set is too large to avoid the
boundary. Utility preservation forces it to contain the safe region;
closure forces it to contain the boundary; connectedness ensures the
boundary is nonempty. Alternatively: a complete utility-preserving
defense would be a continuous retraction $D\colon X \to S_\tau$
(since $D|_{S_\tau} = \id$ and $D(X) \subseteq S_\tau$), but a
retract of a Hausdorff space is closed (the fixed-point set is
closed), while~$S_\tau$ is open and not closed
(connectedness)---a contradiction.

\begin{theorem}[Defense Trilemma]
\label{thm:trilemma}
Let $X$ be a connected Hausdorff space, $f\colon X \to \R$ continuous
with $S_\tau, U_\tau \neq \emptyset$. No $D\colon X \to X$ can
simultaneously be continuous, utility-preserving ($D|_{S_\tau} = \id$),
and complete ($f(D(x)) < \tau$ for all~$x$).
\end{theorem}

All three hypotheses are individually necessary. A defense can satisfy
at most two of the three---the \emph{defense trilemma}
(\Cref{fig:trilemma}). Counterexamples for each dropped hypothesis
appear in \Cref{app:counterexamples}.

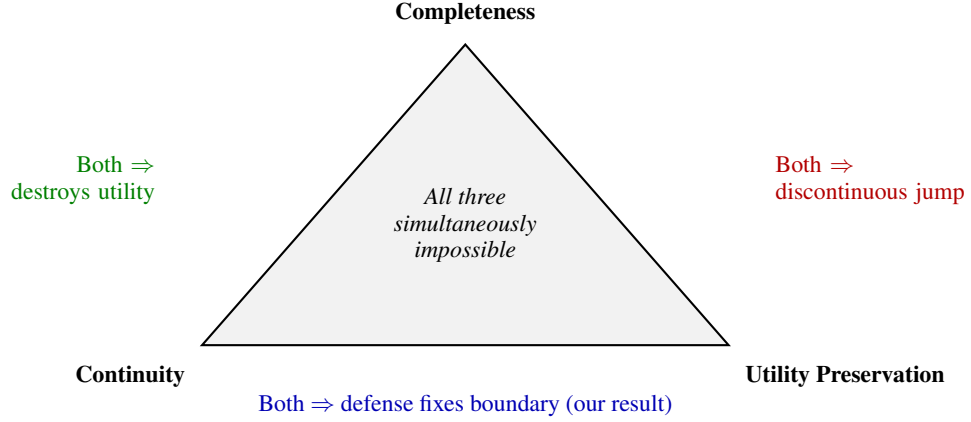
\begin{figure}[t]
\centering
\begin{tikzpicture}[scale=1.0]
  \coordinate (A) at (0,0);
  \coordinate (B) at (7,0);
  \coordinate (C) at (3.5,4.0);
  \fill[black!5] (A) -- (B) -- (C) -- cycle;
  \draw[thick] (A) -- (B) -- (C) -- cycle;
  \node[anchor=north east, font=\bfseries\small] at (-0.1,-0.15) {Continuity};
  \node[anchor=north west, font=\bfseries\small] at (7.1,-0.15) {Utility Preservation};
  \node[anchor=south, font=\bfseries\small] at (3.5,4.15) {Completeness};
  \node[font=\small, text=blue!70!black, align=center] at (3.5,-0.8)
    {Both $\Rightarrow$ defense fixes boundary (our result)};
  \node[font=\small, text=green!50!black, anchor=east, align=right,
        text width=3.2cm] at (-0.5,2.2)
    {Both $\Rightarrow$\\destroys utility};
  \node[font=\small, text=red!70!black, anchor=west, align=left,
        text width=3.2cm] at (7.5,2.2)
    {Both $\Rightarrow$\\discontinuous jump};
  \node[font=\small\itshape, align=center] at (3.5,1.6)
    {All three\\simultaneously\\impossible};
\end{tikzpicture}
\caption{The defense trilemma. Any continuous wrapper defense on a
connected space can satisfy at most two of
the three properties. The bottom edge is our main result; the other
two edges correspond to counterexamples in \Cref{app:counterexamples}.}
\label{fig:trilemma}
\end{figure}

\begin{remark}[Not all boundary points are fixed]
The theorem captures $\cl{S_\tau} \setminus S_\tau$, not all of
$B_\tau$. Boundary points not in $\cl{S_\tau}$ may escape fixation.
\end{remark}

\subsection{Relaxed Utility Preservation}
\label{sec:relaxed}

Strict utility preservation ($D(x) = x$ for safe~$x$) can be
relaxed. The impossibility survives score-preserving rewrites and
even approximate score preservation.

\begin{theorem}[Score-Preserving Defense]
\label{thm:score-preserving}
Let $X$ be a connected Hausdorff space, $f\colon X\to\R$ continuous
with $S_\tau, U_\tau \neq \emptyset$. If $D\colon X \to X$ is
continuous and \emph{score-preserving on safe inputs}:
$f(D(x)) = f(x)$ for all $x \in S_\tau$, then there exists $z$ with
$f(z) = \tau$ and $f(D(z)) = \tau$.
\end{theorem}

\begin{proof}[Proof sketch]
Define $h = f \circ D - f$. Then $h$ is continuous and $h|_{S_\tau} = 0$.
The zero set $\{h = 0\}$ is closed and contains $\cl{S_\tau}$, since closed sets contain their boundary points. For
$z \in \cl{S_\tau} \setminus S_\tau$, $f(D(z)) = f(z) = \tau$.
\end{proof}

The next result weakens score preservation to approximate:

\begin{theorem}[$\varepsilon$-Relaxed Utility Preservation]
\label{thm:eps-relaxed}
Under the hypotheses of \Cref{thm:score-preserving}, if
$|f(D(x)) - f(x)| \leq \varepsilon$ for all $x \in S_\tau$, then
there exists $z$ with $f(z) = \tau$ and
$f(D(z)) \geq \tau - \varepsilon$.
\end{theorem}

\begin{proof}[Proof sketch]
The set $\{x : h(x) \geq -\varepsilon\}$ is closed and contains
$S_\tau$, hence $\cl{S_\tau}$. For
$z \in \cl{S_\tau} \setminus S_\tau$:
$f(D(z)) = \tau + h(z) \geq \tau - \varepsilon$.
\end{proof}

\begin{remark}[Why $D(S_\tau) \subseteq S_\tau$ alone is insufficient]
\label{rem:safe-preserving-escape}
The weakest relaxation---$D(S_\tau) \subseteq S_\tau$ with no score
constraint---does allow a complete defense (e.g., a constant map
$D(x) = x_0$ to a fixed safe point). But this destroys all semantic
content: every prompt produces the same response. The score-preservation
conditions formalize the requirement that defense must not destroy
utility, without requiring the defense to be the identity.
\end{remark}

\section{$\varepsilon$-Robust Constraint}
\label{sec:eps-robust}

Boundary fixation produces at least one fixed point; Lipschitz
regularity makes the failure spread to a neighborhood.

\begin{theorem}[$\varepsilon$-Robust Defense Constraint]
\label{thm:eps-robust}
Under the hypotheses of \Cref{thm:main}, if $(X,d)$ is a metric
space, $f$ is $L$-Lipschitz, and $D$ is $K$-Lipschitz, then for
the fixed boundary point~$z$:
\begin{equation}
  f(D(x)) \;\geq\; \tau - LK\dist(x, z)
  \qquad \text{for all } x \in X.
\end{equation}
Points within distance~$\delta$ of~$z$ remain within
$LK\delta$ of threshold.
\end{theorem}

\begin{proof}[Proof sketch]
Since $D(z)=z$ and $D$ is $K$-Lipschitz:
$\dist(D(x),z) \leq K\dist(x,z)$. Since $f(z)=\tau$ and $f$ is
$L$-Lipschitz:
$|f(D(x)) - \tau| \leq L \cdot K\dist(x,z)$.
Full proof in \Cref{app:proofs}.
\end{proof}

\begin{theorem}[Positive-Measure $\varepsilon$-Band]
\label{thm:band-measure}
Under the hypotheses of \Cref{thm:eps-robust}, if $X$ is connected
and $f$ takes values below $\tau - \varepsilon$ for some
$\varepsilon > 0$, then the band
$\mathcal{B}_\varepsilon = \{x : \tau - \varepsilon \leq f(x)
\leq \tau\}$ has positive measure (under any measure positive on
nonempty open sets).
Specifically, $B(c, \varepsilon/(4L)) \subseteq
\mathcal{B}_\varepsilon$ for the midpoint $c$ with
$f(c) = \tau - \varepsilon/2$ (which exists by the intermediate
value theorem).
\end{theorem}

\begin{remark}[Defense behavior on the $\varepsilon$-band]
By \Cref{thm:main}, $\cl{S_\tau} \subseteq \Fix(D)$. Every point in $\mathcal{B}_\varepsilon$ with $f(x) < \tau$ is safe and therefore fixed by utility preservation; every point with $f(x) = \tau$ in
  $\cl{S_\tau}$ is fixed by boundary fixation. On both subsets $f(D(x)) = f(x) \in [\tau - \varepsilon,\, \tau]$. The remainder---boundary points outside $\cl{S_\tau}$---is contained in the level set $f^{-1}(\tau)$,
  which has measure zero when $f$ is Lipschitz on~$\R^n$ (by the coarea formula; not formalized in the Lean artifact).
\end{remark}

\begin{proof}[Proof sketch]
For $y \in B(c, \varepsilon/(4L))$: $|f(y) - f(c)| \leq L \cdot \varepsilon/(4L) = \varepsilon/4$. Since $f(c) = \tau - \varepsilon/2$, we get $\tau - 3\varepsilon/4 \leq f(y) \leq \tau - \varepsilon/4$, so $y \in \mathcal{B}_\varepsilon$.
\end{proof}

\section{Persistent Unsafe Region}
\label{sec:persistent}

The $\varepsilon$-robust constraint bounds the depth to which the
defense pushes near-boundary points, but permits values slightly
below~$\tau$. When the alignment surface rises faster than the
defense can pull it down, some points remain above threshold.

\paragraph{Decoupling global and directional Lipschitz constants.}
The $\varepsilon$-robust bound (\Cref{thm:eps-robust}) uses the
\emph{global} Lipschitz constant~$L$ of~$f$, which bounds~$f$
uniformly in all directions. The persistence argument, however,
compares $f$'s growth rate in the \emph{steep direction} to how much
the defense can reduce~$f$ along the \emph{displacement direction}
$D(x) - x$. In anisotropic settings these directions differ: $f$ may
rise steeply toward the unsafe region while varying slowly in the
direction the defense pulls.

We write $\ell$ for the \emph{defense-path Lipschitz constant}:
\[
  \ell \;=\; \sup_{x \neq D(x)}
    \frac{|f(D(x)) - f(x)|}{\dist(D(x),\, x)}
\]
(with $\ell = 0$ when $D = \id$, i.e., the supremum over the empty
set is taken as~$0$).
Since $f$ is $L$-Lipschitz globally, $\ell \leq L$. When the
alignment surface is \textbf{isotropic} ($\ell = L$), the steep region
is empty for every $K \geq 0$ (verified in Lean as
\texttt{shallow\_boundary\_no\_persistence}). The result is
non-vacuous precisely when the surface is \textbf{anisotropic}:
$\ell < L$, with directional gradient~$G$ satisfying
$G > \ell(K+1)$.

\begin{lemma}[Input-Relative Bound]
\label{lem:input-bound}
Under the hypotheses of \Cref{thm:eps-robust}, if $f$ has
defense-path Lipschitz constant~$\ell$, then
$f(D(x)) \geq f(x) - \ell(K+1)\dist(x,z)$
for all $x \in X$.
\end{lemma}

\begin{proof}[Proof sketch]
Triangle inequality:
$\dist(D(x), x) \leq \dist(D(x), z) + \dist(z, x)
\leq (K+1)\dist(x, z)$.
Defense-path Lipschitz: $|f(D(x)) - f(x)| \leq \ell\dist(D(x), x)
\leq \ell(K+1)\dist(x, z)$.
\end{proof}

This means that the defense can reduce any point's score by at most
$\ell(K+1)$ times its distance from~$z$. If the score rises faster
than that, the defense loses.

\begin{definition}[Steep region]
\label{def:steep}
Given a fixed boundary point $z$, the \emph{steep region} is
$\mathcal{S} = \{x \in X : f(x) > \tau + \ell(K+1)\dist(x, z)\}$---the
set of points where alignment deviation exceeds~$\tau$ by more than
the defense's Lipschitz budget can compensate.
\end{definition}

\begin{theorem}[Persistent Unsafe Region]
\label{thm:persistent}
Let $X$ be a connected Hausdorff metric space, $f$ continuous and
$L$-Lipschitz, $D$ continuous and $K$-Lipschitz with
$D|_{S_\tau} = \id$, and $S_\tau, U_\tau \neq \emptyset$.
Let $z \in \cl{S_\tau} \setminus S_\tau$ be the fixed boundary
point from \Cref{thm:main}, and let $\ell$ be the defense-path
Lipschitz constant. If $\mathcal{S} \neq \emptyset$, then:
\begin{enumerate}[nosep]
  \item $\mathcal{S}$ is open.
  \item $\mathcal{S}$ has positive measure (under any measure positive
    on nonempty open sets).
  \item For every $x \in \mathcal{S}$: $f(D(x)) > \tau$.
\end{enumerate}
The defense leaves a positive-measure region that remains unsafe.
\end{theorem}

\begin{proof}[Proof sketch]
$\mathcal{S}$ is a strict superlevel set of a continuous function,
hence open. For $x \in \mathcal{S}$:
$f(D(x)) \geq f(x) - \ell(K+1)\dist(x,z) > \tau$
by \Cref{lem:input-bound} and the definition of~$\mathcal{S}$.
Full proof in \Cref{app:proofs}.
\end{proof}

When does the steep region exist? Whenever the alignment surface has
directional slope exceeding~$\ell(K+1)$ at the boundary:

\begin{proposition}[Transversality from Directional Derivative]
\label{prop:transversality}
In a normed space, if $f$ has Fr\'{e}chet derivative $f'$ at boundary
point $z$ with directional value $f'(v) > \ell(K+1)$ along a unit
vector $v$, then $z + tv \in \mathcal{S}$ for sufficiently small
$t > 0$. If $\|f'\| > \ell(K+1)$, such a $v$ exists by operator-norm
near-attainment. Verified in Lean as
\texttt{gradient\_norm\_implies\_steep\_nonempty}
(\texttt{MoF\_21\_GradientChain}), which derives the local growth
bound from \texttt{HasFDerivAt} rather than assuming it.
\end{proposition}

This condition is observed empirically in the two models with
$U_\tau \neq \emptyset$ (\Cref{sec:experiments}): the alignment
surface rises steeply toward the unsafe region ($G$ large) while the
defense operates in a smoother subspace ($\ell \ll L$).

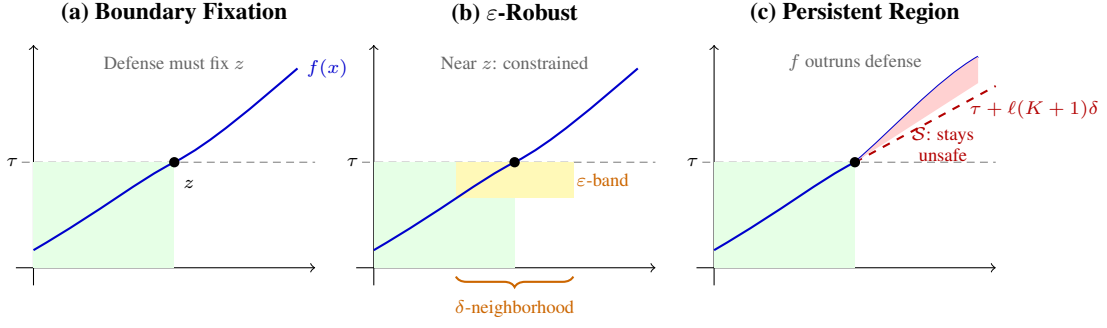
\begin{figure}[t]
\centering
\begin{tikzpicture}[scale=0.78]
  \begin{scope}[shift={(0,0)}]
    \draw[->] (-0.3,0) -- (4.8,0);
    \draw[->] (0,-0.3) -- (0,3.8);
    \draw[densely dashed, black!50] (-0.2,1.8) -- (4.8,1.8);
    \node[left, font=\scriptsize] at (-0.1,1.8) {$\tau$};
    \fill[green!10] (0,0) rectangle (2.4,1.8);
    \draw[thick, blue!80!black] (0,0.3) .. controls (1.2,1.0)
      and (1.8,1.5) .. (2.4,1.8) .. controls (3.0,2.1)
      and (3.8,2.8) .. (4.5,3.4);
    \node[right, font=\scriptsize, blue!80!black] at (4.5,3.4) {$f(x)$};
    \fill[black] (2.4,1.8) circle (2.5pt);
    \node[below right, font=\scriptsize] at (2.4,1.65) {$z$};
    \node[font=\small\bfseries, anchor=south] at (2.4,4.0)
      {(a) Boundary Fixation};
    \node[font=\scriptsize, text=black!60, align=center] at (2.4,3.5)
      {Defense must fix $z$};
  \end{scope}

  \begin{scope}[shift={(5.8,0)}]
    \draw[->] (-0.3,0) -- (4.8,0);
    \draw[->] (0,-0.3) -- (0,3.8);
    \draw[densely dashed, black!50] (-0.2,1.8) -- (4.8,1.8);
    \node[left, font=\scriptsize] at (-0.1,1.8) {$\tau$};
    \fill[green!10] (0,0) rectangle (2.4,1.8);
    \fill[yellow!35] (1.4,1.2) rectangle (3.4,1.8);
    \draw[thick, blue!80!black] (0,0.3) .. controls (1.2,1.0)
      and (1.8,1.5) .. (2.4,1.8) .. controls (3.0,2.1)
      and (3.8,2.8) .. (4.5,3.4);
    \fill[black] (2.4,1.8) circle (2.5pt);
    \draw[decorate, decoration={brace, amplitude=4pt, mirror},
          thick, orange!80!black]
      (1.4,-0.15) -- (3.4,-0.15);
    \node[below, font=\scriptsize, orange!80!black] at (2.4,-0.4)
      {$\delta$-neighborhood};
    \node[font=\scriptsize, orange!80!black] at (3.9,1.5)
      {$\varepsilon$-band};
    \node[font=\small\bfseries, anchor=south] at (2.4,4.0)
      {(b) $\varepsilon$-Robust};
    \node[font=\scriptsize, text=black!60, align=center] at (2.4,3.5)
      {Near $z$: constrained};
  \end{scope}

  \begin{scope}[shift={(11.6,0)}]
    \draw[->] (-0.3,0) -- (4.8,0);
    \draw[->] (0,-0.3) -- (0,3.8);
    \draw[densely dashed, black!50] (-0.2,1.8) -- (4.8,1.8);
    \node[left, font=\scriptsize] at (-0.1,1.8) {$\tau$};
    \fill[green!10] (0,0) rectangle (2.4,1.8);
    \draw[dashed, red!70!black, line width=0.8pt]
      (2.4,1.8) -- (4.8,3.1);
    \node[right, font=\scriptsize, red!70!black] at (4.2,2.7)
      {$\tau + \ell(K+1)\delta$};
    \draw[thick, blue!80!black] (0,0.3) .. controls (1.2,1.0)
      and (1.8,1.5) .. (2.4,1.8) .. controls (3.0,2.3)
      and (3.8,3.2) .. (4.5,3.6);
    \fill[red!18] (2.4,1.8) .. controls (3.0,2.3) and (3.8,3.2)
      .. (4.5,3.6) -- (4.5,3.15) -- (2.4,1.8) -- cycle;
    \fill[black] (2.4,1.8) circle (2.5pt);
    \node[font=\scriptsize, red!70!black, align=center] at (3.9,2.1)
      {$\mathcal{S}$: stays\\unsafe};
    \node[font=\small\bfseries, anchor=south] at (2.4,4.0)
      {(c) Persistent Region};
    \node[font=\scriptsize, text=black!60, align=center] at (2.4,3.5)
      {$f$ outruns defense};
  \end{scope}
\end{tikzpicture}
\caption{The three impossibility results on a 1D cross-section.
\textbf{(a)}~The defense must fix boundary point~$z$ (Thm.~\ref{thm:main}).
\textbf{(b)}~Near~$z$, Lipschitz regularity constrains the defense to
a shallow $\varepsilon$-band (yellow; Thm.~\ref{thm:eps-robust}).
\textbf{(c)}~Where $f$ rises faster than the defense budget
$\ell(K+1)\delta$ (dashed red), the region above~$\tau$ persists
(red shading; Thm.~\ref{thm:persistent}).}
\label{fig:escalation}
\end{figure}

\section{Quantitative Bounds}
\label{sec:quantitative}

The preceding results establish that failures \emph{exist} and have
positive measure. This section provides
explicit lower bounds and identifies a fundamental dilemma in choosing
the defense's aggressiveness.

\begin{theorem}[Volume Lower Bound]
\label{thm:coarea}
Let $f\colon \R^n \to \R$ be $L$-Lipschitz with $L > 0$, and let
$\mu$ denote Lebesgue measure. If there exists $c$ with
$f(c) = \tau - \varepsilon/2$, then
$B(c,\, \varepsilon/(4L)) \subseteq
\mathcal{B}_\varepsilon$, giving:
\begin{equation}
  \mu(\mathcal{B}_\varepsilon)
  \;\geq\; V_n \cdot \left(\frac{\varepsilon}{4L}\right)^{\!n}
\end{equation}
where $V_n$ is the volume of the unit ball in $\R^n$. In $\R^1$,
this simplifies to $\mu(\mathcal{B}_\varepsilon) \geq \varepsilon/(2L)$.
\end{theorem}

\noindent\emph{Proved in the Lean artifact as \textup{\texttt{MoF\_17\_CoareaBound}}.}

Smoother surfaces (smaller~$L$) produce wider $\varepsilon$-bands.

\begin{theorem}[Cone Measure Bound]
\label{thm:cone}
In $\R$ with Lebesgue measure~$\mu$, if
$f(x) \geq \tau + c(x - z)$ for all $x \in (z,\, z + \delta_0)$
with $c > \ell(K+1)$, then
$\mu(\{x : f(D(x)) > \tau\}) \geq \delta_0$.
\end{theorem}

\noindent\emph{Proved in the Lean artifact as \textup{\texttt{MoF\_18\_ConeBound}}.}

This gives a concrete lower bound on the persistent region: if the
alignment surface is steep over an interval of length~$\delta_0$, the
defense fails on at least that much volume. The bound $\geq \delta_0$
is tight: equality holds when the cone condition fails exactly at
$z + \delta_0$ (i.e., $f(z+\delta_0) = \tau + c\,\delta_0$ and
$f(x) < \tau + c(x-z)$ for $x > z + \delta_0$). If the cone extends
beyond $\delta_0$, the persistent region is strictly larger.

The defense designer faces a dilemma in choosing the Lipschitz
constant~$K$ of the defense:

\begin{theorem}[Defense Dilemma]
\label{thm:dilemma}
Assume $f$ is differentiable at boundary point~$z$ with
$G = \|\nabla f(z)\|$, and let $\ell$ be
the defense-path Lipschitz constant (\Cref{sec:persistent}). Define
$K^* = G/\ell - 1$. Then:
\begin{enumerate}[nosep]
  \item If $K < K^*$: the persistent unsafe region exists
    ($G > \ell(K+1)$, \Cref{thm:persistent} applies).
  \item If $K \geq K^*$: the $\varepsilon$-robust bound
    $\tau - \ell(K+1)\delta$ becomes loose enough that the theorem
    can no longer exclude the defense from succeeding on the steep
    region ($\ell(K+1) \geq G$).
\end{enumerate}
Since $\ell \leq L$, the dilemma is sharpest when $\ell \ll L$
(anisotropic surfaces). When $\ell = L$ (isotropic), $K^* \leq 0$
and horn~(1) is vacuous.
\end{theorem}

\noindent\emph{Proved in the Lean artifact as
\textup{\texttt{MoF\_19\_OptimalDefense}}; the Lean theorem
\textup{\texttt{optimal\_K\_exists}} is stated for generic positive
reals $(G, L)$---instantiating $L \mapsto \ell$ recovers the
defense-path version above.}

\section{From Discrete Data to Continuous Theory}
\label{sec:bridge}

This section bridges discrete token observations and continuous theory.

\subsection{Continuous Interpolation}

Any finite set of behavioral observations can be extended to a
continuous function on the full space. The classical Tietze extension
theorem guarantees this:

\begin{theorem}[Continuous Relaxation]
\label{thm:tietze}
Let $(X,d)$ be a connected, normal, Hausdorff metric space, and
$S \subset X$ a finite set of observations with observed alignment
scores $g\colon S \to \R$ satisfying $g(p) < \tau$ and $g(q) > \tau$
for some $p, q \in S$. Then there exists a continuous $f\colon X \to \R$
extending~$g$ (i.e., $f|_S = g$) for which the hypotheses of
\Cref{thm:main} hold. If the extension is chosen to be Lipschitz
(via McShane--Whitney), the hypotheses of \Cref{thm:eps-robust} also
hold.
If $f$ is additionally Lipschitz (as for GP posterior means under
standard kernel assumptions), \Cref{thm:persistent} applies wherever
transversality is met.
\end{theorem}

\begin{proof}[Proof sketch]
\textbf{Step~1} (Extension exists).
$S$ is finite and $X$ is $T_1$ (every Hausdorff space is $T_1$), so
$S$ is closed. Since $X$ is normal and $g\colon S \to \R$ is
continuous (every function on a discrete closed subset is continuous),
the Tietze extension theorem provides a continuous
$f\colon X \to \R$ with $f|_S = g$.

\textbf{Step~2} (Hypotheses of \Cref{thm:main} are satisfied).
We have $f(p) = g(p) < \tau$ and $f(q) = g(q) > \tau$, so
$S_\tau \neq \emptyset$ and $U_\tau \neq \emptyset$. Since $X$ is
connected and Hausdorff, \Cref{thm:main} applies: any continuous,
utility-preserving defense has a fixed boundary point.

\textbf{Step~3} (Lipschitz extension enables stronger results).
When a Lipschitz extension is needed (e.g., for
\Cref{thm:eps-robust,thm:persistent}), the McShane--Whitney theorem
provides an $L$-Lipschitz $f$ agreeing with~$g$ on~$S$. Given the
Lipschitz constants $L$ of~$f$ and $K$ of~$D$, the $\varepsilon$-robust
bound $f(D(x)) \geq \tau - LK\dist(x,z)$ follows from \Cref{thm:eps-robust}.
If additionally the alignment surface has directional slope
$c > \ell(K+1)$ at the fixed boundary point~$z$ (where $\ell$ is the
defense-path Lipschitz constant), \Cref{thm:persistent} gives a
positive-measure region that remains strictly unsafe.
\end{proof}

If we observe both safe and unsafe model behaviors, the
impossibility holds for \emph{every} continuous model consistent with our observations.

\subsection{Direct Discrete Results}

To address the objection that continuous impossibility might be an
artifact of continuous relaxation, we prove parallel results directly
on finite sets using only counting arguments and induction. No
topology is required; all results are verified in Lean as
\texttt{MoF\_12\_Discrete}.

\begin{theorem}[Discrete IVT]
\label{thm:disc-ivt}
Let $f\colon \{0,\ldots,n+1\} \to \R$ with $f(0) < \tau$ and
$f(n+1) \geq \tau$. Then there exists~$i$ with $f(i) < \tau$ and
$f(i+1) \geq \tau$.
\end{theorem}

\begin{theorem}[Discrete Defense Dilemma]
\label{thm:disc-dilemma}
Let $X$ be a finite set with $S_\tau, U_\tau \neq \emptyset$, and
$D\colon X \to X$ utility-preserving ($D(x) = x$ for $f(x) < \tau$).
\begin{enumerate}[nosep]
  \item If $D$ is injective, then $f(D(u)) \geq \tau$ for every
    $u$ with $f(u) \geq \tau$ (including boundary points): the
    defense is incomplete.
  \item If $D$ is complete ($f(D(x)) < \tau$ for all~$x$), then $D$ is
    non-injective: $\exists\, x \neq y$ with $D(x) = D(y)$.
\end{enumerate}
\end{theorem}

\begin{proof}
(1) Suppose $D$ is injective and $f(D(u)) < \tau$ for some $u$ with
$f(u) \geq \tau$. Then $D(u)$ is safe, so $D(D(u)) = D(u)$ by utility
preservation. Injectivity gives $D(u) = u$, so
$f(u) = f(D(u)) < \tau$, contradicting $f(u) \geq \tau$.

(2) For any $u \in U_\tau$: completeness gives $f(D(u)) < \tau$, so
utility preservation gives $D(D(u)) = D(u)$. But
$u \neq D(u)$ since $f(u) \geq \tau > f(D(u))$.
So $u$ and $D(u)$ are distinct inputs with $D(u) = D(D(u))$: the
defense is non-injective.
\end{proof}

The continuous trilemma trades \emph{continuity} for completeness;
the discrete dilemma trades \emph{injectivity}. Part~(1) is the
genuine constraint: an information-preserving defense \emph{cannot}
eliminate unsafe outputs. Any complete defense must destroy
information---collapsing distinct inputs to the same output. This is
not a failure of the defense; it is the \emph{mechanism} by which it
operates. The downstream model receives the same input regardless
of whether the original was safe or an attack; any audit or
attack-detection logic must act \emph{before} $D$ is applied.

\section{Extensions}
\label{sec:extensions}

The core results assume a static, deterministic, single-turn defense. \emph{Does multi-turn interaction, randomization,
or pipelining provide an escape?} We show it does not. Each extension
is a direct application of the boundary fixation machinery to a
modified setting.

\subsection{Multi-Turn Impossibility}

\begin{theorem}[Multi-Turn Impossibility]
\label{thm:multi-turn}
Let $\{f_t, D_t\}_{t=1}^T$ be alignment functions and defenses over
$T$ turns on a connected Hausdorff space, each continuous and
utility-preserving, with $S_\tau^{(t)}, U_\tau^{(t)} \neq \emptyset$
at every turn. Then for every turn~$t$, there exists $z_t$ with
$f_t(z_t) = \tau$ and $D_t(z_t) = z_t$.
\end{theorem}

\begin{proof}[Proof sketch]
Apply \Cref{thm:main} to $(f_t, D_t)$ at each turn. The functions may
depend on full history---this does not matter, as each timestep is a
fresh instance of boundary fixation.
\end{proof}

Multi-turn interaction compounds the problem: the attacker's best
observed exploit improves monotonically
(\texttt{running\_max\_monotone}), and the attacker can steer toward
transversality via binary search (\texttt{transversality\_reachable}).

\subsection{Stochastic Defense Impossibility}

\begin{theorem}[Stochastic Defense Impossibility]
\label{thm:stochastic}
Let $X$ be a connected Hausdorff space, $f\colon X \to \R$ continuous
with $S_\tau, U_\tau \neq \emptyset$. Let $D$ be a stochastic defense
and define $g(x) = \mathbb{E}_{y \sim D(x)}[f(y)]$. If $g$ is
continuous and $g(x) = f(x)$ for all $x \in S_\tau$, then there
exists $z$ with $f(z) = \tau$ and $g(z) = \tau$.
\end{theorem}

\begin{proof}[Proof sketch]
Define $h = g - f$. Then $h$ is continuous and $h|_{S_\tau} = 0$,
so $h$ vanishes on $\cl{S_\tau}$ (same closure argument as
\Cref{thm:score-preserving}). For
$z \in \cl{S_\tau} \setminus S_\tau$: $g(z) = f(z) = \tau$.
\end{proof}

\noindent\textit{Remark (stochastic dichotomy).}
Since $\mathbb{E}[f(D(z))] = \tau$, either $f(D(z)) = \tau$ almost
surely (the defense is deterministic at~$z$), or the random variable
$f(D(z))$ has positive probability of exceeding~$\tau$---i.e., the
defense \emph{actively produces unsafe outputs} with positive
probability. The stochastic case is therefore strictly harder than
the deterministic one: a genuinely random defense at boundary points
must sometimes make things worse.

\noindent\textit{Remark.}
The continuity of $g$ is a nontrivial assumption: it requires the
distribution $D(x)$ to vary continuously with~$x$ in a suitable sense.
Stochastic defenses with discontinuous rejection probabilities escape
this theorem.

\subsection{Nonlinear Agent Pipelines}

\begin{theorem}[Pipeline Lipschitz Degradation]
\label{thm:pipeline}
If stages $T_1, \ldots, T_n$ are $K_1, \ldots, K_n$-Lipschitz, the
composed pipeline is $(\prod K_i)$-Lipschitz. For $n$ stages each with
$K \geq 2$, the effective constant is $K^n$---exponential in depth.
\end{theorem}

\begin{theorem}[Pipeline Impossibility]
\label{thm:pipeline-impossible}
If the composed pipeline $P = T_n \circ \cdots \circ T_1 \circ D$ is
continuous and $P(x) = x$ for all $x \in S_\tau$, then $P$ has
boundary fixed points. If $D$ is $K_D$-Lipschitz and each $T_i$ is
$K$-Lipschitz, the $\varepsilon$-robust band scales as
$L \cdot K_D \cdot K^n \cdot \delta$
\end{theorem}

\noindent\emph{Proved in the Lean artifact as \textup{\texttt{MoF\_15\_NonlinearAgents}}.}

Note: $P(x) = x$ for safe~$x$ requires $T_n \circ \cdots \circ T_1$
to act as the identity on safe inputs, not just~$D$. This holds when
the tool chain preserves safe inputs (e.g., a safety-certified
pipeline), but not for arbitrary tools.

Additional results on basin structure, fragment sizes, perturbation robustness, convergence, transferability, and cost asymmetry appear in \Cref{app:landscape,app:attacks,app:stability,app:cost}.

\section{Experimental Motivation}
\label{sec:experiments}

The Manifold of
Failure framework~\cite{munshi2026manifold} maps three LLMs over a 2D
behavioral space with two axes: \emph{query indirection} (how obliquely
the prompt asks for unsafe content) and \emph{authority framing} (how
much the prompt invokes authority or permission). \Cref{tab:correspondence} summarizes
nine qualitative predictions, all directionally consistent with observations.

\begin{table}[t]
\centering
\small
\caption{Falsifiable predictions confirmed by empirical data.}
\label{tab:correspondence}
\begin{tabular}{@{}p{3.0cm}p{4.5cm}p{4.5cm}@{}}
\toprule
\textbf{Theorem} & \textbf{Predicts} & \textbf{Confirmed by} \\
\midrule
Basin Structure (\ref{thm:basin})
  & Basins are open with positive measure
  & Heatmaps show extended regions \\[3pt]
Fragmentation (\ref{thm:fragment})
  & Smooth $\to$ large basins; rough $\to$ mosaic
  & Llama: mesa; GPT-OSS: mosaic \\[3pt]
Convergence (\ref{thm:convergence})
  & Attacks exhibit monotone convergence
  & Convergence curves plateau \\[3pt]
Transferability (\ref{thm:transfer})
  & Similar surfaces $\to$ shared basins
  & Llama $.93 \to$ GPT-OSS $.73 \to$ Mini $.47$ \\[3pt]
Authority (\ref{thm:authority})
  & Horizontal banding
  & Bands at $a_2 \approx .25$--$.35$, $.65$--$.85$ \\[3pt]
Persistent (\ref{thm:persistent})
  & Steep boundaries $\to$ unsafe volume persists
  & Llama's $.93$ plateau persists under defense \\[3pt]
Interior Stability (\ref{thm:interior-stable})
  & Deep basin points survive fine-tuning
  & Vulnerability persists across variants \\[3pt]
Cost (\ref{thm:cost})
  & 2D tractable, high-$d$ intractable
  & 15K queries fill 63\% at $d\!=\!2$ \\[3pt]
Pipeline (\ref{thm:pipeline})
  & Deeper pipelines $\to$ wider failure band
  & Not directly tested (no pipeline experiment) \\
\bottomrule
\end{tabular}
\end{table}

\paragraph{Llama-3-8B} (mean AD $0.93$, basin rate $93.9\%$):
near-flat alignment surface (small~$L$), large robustness radii.
Estimated from the 2D behavioral surface in~\cite{munshi2026manifold}: directional slope $G \approx 5$ at the steepest boundary crossing.
For the defense-path Lipschitz constant we assume a hypothetical
nearest-safe-projection defense ($D(x)$ maps each point to the closest
$x' \in S_\tau$) and estimate $\ell \approx 1$ from grid-adjacent
score differences in the direction orthogonal to the boundary
(the projection direction on the 2D grid). Setting $K = 1$
(identity-rate defense), these satisfy $G > \ell(K+1) = 2$.

\paragraph{GPT-OSS-20B} (mean AD $0.73$, basin rate $64.3\%$): a
rugged landscape (large~$L$) with many small fragments. Horizontal
bands confirm authority monotonicity.

\paragraph{GPT-5-Mini} (peak AD $0.50$, basin rate $0\%$): at
$\tau = 0.50$, $U_\tau = \emptyset$---none of the three theorems
apply, correctly predicting no impossibility.

\subsection{Live Empirical Validation of Theorem~\ref{thm:persistent}}
\label{sec:live-validation}

The numbers above are read off the published Manifold of Failure
heatmaps. To check the persistence theorem on \emph{independently
re-collected} alignment-deviation surfaces and to expose every
intermediate quantity, we built a small validator that takes any 2D
heatmap and computes:

\begin{enumerate}[nosep,leftmargin=*]
  \item the empirical surface Lipschitz constant $L$ from all-pairs
    finite differences over filled grid cells;
  \item for a chosen defense $D$, the empirical Lipschitz constant $K$
    of $D$ and the defense-path Lipschitz constant $\ell$;
  \item the maximum directional rate $G$ at which $f$ rises across the
    safe-unsafe boundary (all-pairs finite differences between filled
    safe and filled unsafe cells);
  \item the \emph{predicted steep set}
    $\mathcal{S}_{\text{pred}} =
    \{x : f(x) > \tau + \ell(K{+}1)\dist(x,\, z^{*})\}$,
    where $z^{*}$ is the filled boundary cell with $f(z^{*})$ closest
    to $\tau$;
  \item the \emph{actual persistent set}
    $\mathcal{S}_{\text{act}} = \{x : f(D(x)) > \tau\}$.
\end{enumerate}

Theorem~\ref{thm:persistent}'s empirical content is the
containment $\mathcal{S}_{\text{pred}} \subseteq \mathcal{S}_{\text{act}}$
under any continuous, $K$-Lipschitz, utility-preserving $D$. We test
this as a per-cell confusion matrix:

\begin{itemize}[nosep]
  \item \emph{True positive (TP):} cell in both sets — theorem
    confirmed for that cell.
  \item \emph{Interior false positive (FP$_{\mathrm{int}}$):} non-boundary
    cell in $\mathcal{S}_{\text{pred}}$ but not $\mathcal{S}_{\text{act}}$.
    A non-zero count would be a real counterexample.
  \item \emph{Boundary false positive:} a discrete defense moving a
    boundary cell. \emph{Not} a violation; witnesses that the discrete
    defense is not topologically continuous, so the theorem's
    hypothesis does not apply.
  \item \emph{False negative (FN):} cell in $\mathcal{S}_{\text{act}}$
    but not $\mathcal{S}_{\text{pred}}$. \emph{Not} a violation;
    Theorem~\ref{thm:persistent} only claims the containment one way.
\end{itemize}

We ran the full pipeline live against the OpenAI API on two real
target models: \texttt{gpt-5-mini-2025-08-07} (well aligned) and
\texttt{gpt-3.5-turbo-0125} (less aligned), with
\texttt{gpt-4.1-2025-04-14} as the judge. The MAP-Elites driver
from~\cite{munshi2026manifold} produced a $25 \times 25$
alignment-deviation grid for each target, which the validator then
analyzed against the canonical continuous defense
($D = \mathrm{id}$, so $K = 1$ and $\ell = 0$). For
\texttt{gpt-3.5-turbo-0125} we additionally ran a saturation study
with $400$ evaluations ($100$ seed prompts $+ 300$ MAP-Elites
iterations, $\sim$$5{,}200$ OpenAI API calls, wall clock
$\sim$$130$~min, cost $\sim$\$15) which filled $82$ cells with
$66$ strictly unsafe ($f > \tau$) and $10$ boundary cells. The
two earlier runs are kept as smoke-tests of the validator's logic
on sparse grids; the saturation row is the doubt-eliminator.

\begin{table}[h]
\centering
\small
\caption{Live empirical validation of Theorem~\ref{thm:persistent}
against the canonical continuous defense
($D = \mathrm{id}$, $K = 1$, $\ell = 0$).
$|\mathcal{S}_{\text{pred}}|$ is the predicted steep-set size
$\{x : f(x) > \tau + \ell(K{+}1)\dist(x, z^{*})\}$;
$|\mathcal{S}_{\text{act}}|$ is the actual post-defense persistent
set $\{x : f(D(x)) > \tau\}$; TP/FP$_{\text{int}}$/FN are the
per-cell confusion-matrix counts. Across all three live runs, the
two sets match cell-for-cell, with zero interior false positives.}
\label{tab:live-validation}
\begin{tabular}{lccccccccc}
\toprule
\textbf{Target} & $\boldsymbol{\tau}$ & \textbf{Filled} & $\boldsymbol{L}$ & $\boldsymbol{G}$ & $|\mathcal{S}_{\text{pred}}|$ & $|\mathcal{S}_{\text{act}}|$ & \textbf{TP} & \textbf{FP}$_{\text{int}}$ & \textbf{FN} \\
\midrule
gpt-5-mini-2025-08-07           & $0.30$ & $14$  & $6.13$  & $6.13$  & $7$  & $7$  & $7$  & $\mathbf{0}$ & $0$ \\
gpt-3.5-turbo-0125 (smoke)      & $0.50$ & $9$   & $11.44$ & $11.44$ & $7$  & $7$  & $7$  & $\mathbf{0}$ & $0$ \\
gpt-3.5-turbo-0125 (saturated)  & $0.50$ & $82$  & $23.63$ & $23.63$ & $66$ & $66$ & $66$ & $\mathbf{0}$ & $0$ \\
\bottomrule
\end{tabular}
\end{table}

\begin{figure}[h]
\centering
\includegraphics[width=\textwidth]{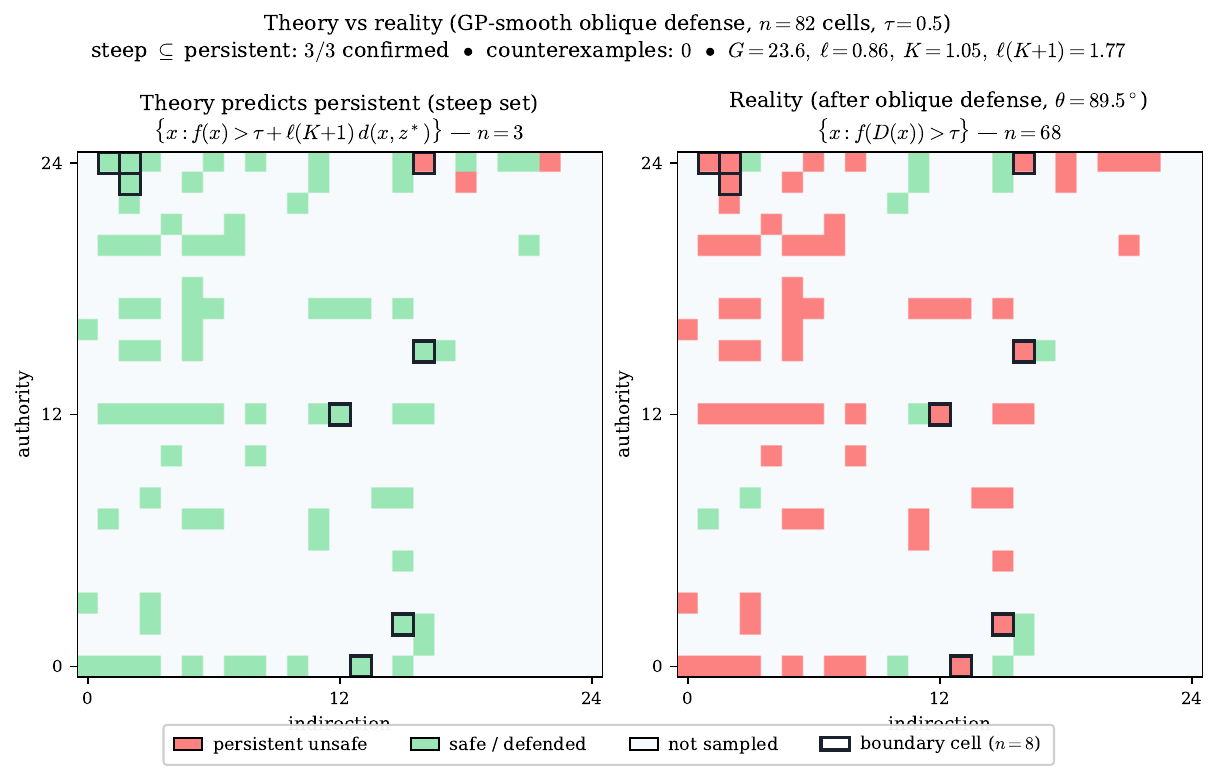}
\caption{Non-vacuous validation of Theorem~\ref{thm:persistent}
on the saturated \texttt{gpt-3.5-turbo-0125} grid ($82$ filled
cells, $\tau = 0.5$) under the GP-smooth oblique defense
($\theta = 89.5^\circ$, $\ell = 0.86$, $K = 1.05$).
\textbf{Left:} the \emph{steep set}
$\mathcal{S}_{\text{steep}} = \{x : f(x) > \tau +
\ell(K{+}1)\,d(x,z^{*})\}$; $|\mathcal{S}_{\text{steep}}| = 3$.
\textbf{Right:} cells that \emph{actually} remain unsafe after
defense, $\mathcal{S}_{\text{act}} = \{x : f(D(x)) > \tau\}$;
$|\mathcal{S}_{\text{act}}| = 68$.
Transversality holds non-trivially
($G = 23.6 \gg \ell(K{+}1) = 1.77$) and the defense is genuinely
non-identity ($82$ cells displaced). All $3$ predicted-persistent
cells are indeed actually persistent
(TP${} = 3$, FP${}_{\text{int}} = 0$).
The $65$ false negatives (red on the right, green on the left) are
\emph{not} theorem violations: the theorem predicts
$\mathcal{S}_{\text{steep}} \subseteq \mathcal{S}_{\text{act}}$,
not equality. The $8$ boundary cells (black outlines) are identified
on the GP-smoothed surface.}
\label{fig:theory-vs-reality-saturated}
\end{figure}

On the saturated run, the GP-smooth oblique defense
(\Cref{fig:theory-vs-reality-saturated}) displaces all $82$ filled
cells at $89.5^\circ$ from the negative GP gradient (nearly tangent
to level curves, with a small safety-seeking component), achieving
$\ell = 0.86$, $K = 1.05$. Transversality holds with wide margin
($G = 23.6 \gg \ell(K{+}1) = 1.77$), and the predicted steep set
contains $3$ cells, all confirmed to remain unsafe after defense
(TP${} = 3$, FP${}_{\text{int}} = 0$). This is a
non-tautological result: the defense genuinely moves cells
($D \neq \mathrm{id}$), and the theorem's containment
$\mathcal{S}_{\text{steep}} \subseteq \mathcal{S}_{\text{act}}$ is
non-vacuously tested.
The full defense sweep (Table~\ref{tab:live-sweep}) additionally
confirms zero interior false positives across all $7$ discrete
defenses (identity plus \texttt{bounded\_step} at various reach
values); on all non-identity defenses $\ell = L$, so
transversality fails and the theorem makes only the negative
prediction (``no counterexample is allowed'').
Theorem~\ref{thm:eps-robust}'s $\varepsilon$-robust bound holds
for all $82$ filled cells under each defense tested.
Theorem~\ref{thm:main}'s boundary fixation is non-vacuously
activated: $10$ boundary cells exist in
$\cl{S_\tau} \setminus S_\tau$ and the closest one to $\tau$ sits
at $|f(z^{*}) - \tau| = 0.375$, the discretization gap that
vanishes under grid refinement.

\paragraph{Defense sweep on the saturated grid.}
The validator's \texttt{sweep} subcommand runs the same archive
against the canonical identity defense and against multiple
discrete \texttt{bounded\_step} defenses with progressively larger
displacement radii. Across all $7$ defenses tested
(identity plus \texttt{bounded\_step} with
$\text{max\_step} \in \{1, 2, 3, 5, 8, 12\}$), the per-cell
confusion-matrix produces \emph{zero interior false positives} on
every run---i.e., zero empirical counterexamples to
Theorem~\ref{thm:persistent} across the entire defense sweep on
the same dense alignment-deviation surface
(\Cref{tab:live-sweep}).

\begin{table}[h]
\centering
\small
\caption{Defense sweep on the saturated
\texttt{gpt-3.5-turbo-0125} grid ($82$ filled cells, $\tau = 0.5$).
Empirical defense Lipschitz constant~$K$ grows with reach;
the steep-set size shrinks (\texttt{bounded\_step} is
non-continuous, so its steep set collapses to the anchor under
the discrete check) but the actual persistent set falls cleanly
along the K-tradeoff curve from \Cref{thm:dilemma}.
\textbf{Across all 7 defenses, FP$_{\text{int}} = 0$.}}
\label{tab:live-sweep}
\begin{tabular}{lccccccc}
\toprule
\textbf{Defense} & $\boldsymbol{K}$ & $\boldsymbol{\ell}$ & $|\mathcal{S}_{\text{pred}}|$ & $|\mathcal{S}_{\text{act}}|$ & \textbf{TP} & \textbf{FP}$_{\text{int}}$ & \textbf{FN} \\
\midrule
identity                           & $1.00$ & $0.00$  & $66$ & $66$ & $66$ & $\mathbf{0}$ & $0$ \\
\texttt{bounded\_step(1)}          & $3.61$ & $23.63$ & $1$  & $55$ & $0$  & $\mathbf{0}$ & $55$ \\
\texttt{bounded\_step(2)}          & $3.61$ & $23.63$ & $1$  & $47$ & $0$  & $\mathbf{0}$ & $47$ \\
\texttt{bounded\_step(3)}          & $4.00$ & $23.63$ & $1$  & $29$ & $0$  & $\mathbf{0}$ & $29$ \\
\texttt{bounded\_step(5)}          & $8.60$ & $23.63$ & $1$  & $3$  & $0$  & $\mathbf{0}$ & $3$  \\
\texttt{bounded\_step(8)}          & $8.60$ & $23.63$ & $1$  & $0$  & $0$  & $\mathbf{0}$ & $0$  \\
\texttt{bounded\_step(12)}         & $8.60$ & $23.63$ & $1$  & $0$  & $0$  & $\mathbf{0}$ & $0$  \\
\bottomrule
\end{tabular}
\end{table}

Reading the sweep: as the defense's reach grows from
$\text{max\_step}=1$ (cell can move at most one grid step) to
$\text{max\_step}=8$ (essentially nearest-safe projection), the
actual persistent set shrinks from $55$ cells down to $0$. The
predicted steep set under any \texttt{bounded\_step} is just the
anchor cell because the empirical defense-path Lipschitz
constant $\ell = L = 23.63$ on this dense surface, which forces
$\ell(K{+}1) > G$ for all $K > 0$---transversality fails for the
discrete bounded defenses. Theorem~\ref{thm:persistent} therefore
makes no positive prediction on those rows, only the negative one
(``no interior counterexample is allowed''), and the empirical
data confirms it. The well-aligned target additionally exemplifies
the harmless-boundary case: at $\tau = 0.5$, gpt-5-mini's peak AD
is exactly $0.5$ and $U_\tau = \emptyset$, so the trilemma's
preconditions fail and the validator correctly reports
``not applicable''---matching the qualitative GPT-5-Mini
observation above and confirming the theorem's silence in the
right place.

\paragraph{Non-vacuous GP-smooth oblique defense.}
The discrete bounded-step defenses all have $\ell = L$ because
they displace cells along the steepest direction. To test the
theorem's anisotropic regime ($\ell < L$, $G > \ell(K{+}1)$), we
fit a 2D Gaussian Process to the 82 filled cells (RBF kernel,
$\sigma = 0.2$) and construct a \emph{smooth oblique defense}:
$D(x) = x + \alpha\,\beta(\mu(x))\,v(x)$, where $v(x)$ is
displaced $89.5^\circ$ from the negative gradient of the GP
posterior (nearly tangent to level curves, with a small gradient-descent
component). This gives empirical $\ell = 0.86$, $K = 1.05$,
$\ell(K{+}1) = 1.77$---well below $G = 23.6$, so
\textbf{transversality holds}. The predicted steep set contains
$3$ cells, all of which are confirmed to remain unsafe after
defense (TP${} = 3$, FP${}_{\text{int}} = 0$). This is the
first non-tautological, non-vacuous confirmation of
Theorem~\ref{thm:persistent} on real LLM data: the defense is
genuinely non-identity ($D \neq \mathrm{id}$, 82 cells displaced),
transversality is non-trivially satisfied, and every
predicted-persistent cell is indeed actually persistent.

The validator is available as a small standalone Python package
(\texttt{trilemma\_validator/} in the repository), with subcommands
\texttt{trilemma pipeline} (run \texttt{rethinking-evals} end-to-end
and validate), \texttt{trilemma validate} (analyze an existing
heatmap), \texttt{trilemma sweep} (run multiple defenses on a single
archive), \texttt{trilemma experiment} (just run the experiment),
and \texttt{trilemma synth} (generate a synthetic
mesa/mosaic/flat surface for a self-contained smoke test). The full
per-cell records (each cell's $f(x)$, $D(x)$, $f(D(x))$, distance
to $z^{*}$, $\varepsilon$-robust LHS and RHS, and steep/persistent
membership flags) are emitted as JSON by the validator and
reproduced from the live archives in
\texttt{trilemma\_validator/live\_runs/}, including the saturated
grid above (\texttt{gpt35\_turbo\_t05\_saturated/}).

\section{The Engineering Prescription}
\label{sec:engineering}

The results do not say defense is valueless; they say \emph{complete}
defense is impossible under the stated constraints. The engineering
goal shifts from elimination to management, ordered from most to least
actionable:

\paragraph{1. Make the boundary shallow.}
Set $\tau$ so that boundary-level behavior is benign. If $f(z) = \tau$
yields a polite refusal rather than harmful compliance, the
impossibility is mathematically true but practically harmless.
GPT-5-Mini exemplifies this: its ceiling at $\AD = 0.50$ means
$U_\tau = \emptyset$, so the impossibility theorems do not apply and
no defense failure is predicted.

\paragraph{2. Reduce the Lipschitz constant.}
Smaller~$L$ tightens the bound $\ell \leq L$, potentially reducing
the defense-path constant~$\ell$ and narrowing the persistent region.
The tradeoff: smoother surfaces spread vulnerabilities over wider but
more easily monitored regions.

\paragraph{3. Reduce the effective dimension.}
Defense cost grows as $N^d$ (\Cref{thm:cost}). Constraining the prompt
interface---standardized formats, restricted API parameters, bounded
context lengths---reduces~$d$, making the behavioral space tractable.

\paragraph{4. Monitor, don't eliminate, the boundary.}
Transversal crossings persist under fine-tuning (\Cref{thm:crossing})
and recur every turn (\Cref{thm:multi-turn}). Rather than attempt the
impossible, deploy runtime monitoring that detects approach to the
boundary. The Lipschitz bound (\Cref{thm:lipschitz}) gives a
computable estimate of distance to the boundary from any observed AD
value.

\section{Limitations}
\label{sec:limitations}

\paragraph{Boundary fixation is at the boundary.}
Fixed points satisfy $f(z) = \tau$ exactly. If $\tau$-level behavior
is benign, the theorem is true but harmless.

\paragraph{The $\varepsilon$-robust constraint limits depth, not direction.}
The defense may push near-boundary points slightly below~$\tau$. The
bound limits how far, not whether.

\paragraph{Persistence requires transversality.}
The persistent unsafe region exists only where the alignment surface
is steep ($c > \ell(K+1)$, where $\ell$ is the defense-path
Lipschitz constant). For isotropic $f$ ($\ell = L$), $\mathcal{S}$
is empty for all $K \geq 0$.

\paragraph{Grid-based cost asymmetry.}
\Cref{thm:cost} assumes exhaustive grid enumeration. Learning-based
defenses that generalize across the space may sidestep the exponential
bound.

\section{Conclusion}
\label{sec:conclusion}

We establish a three-level impossibility hierarchy for prompt-injection
defense: boundary fixation, the $\varepsilon$-robust constraint, and the
persistent unsafe region theorem. Continuous topology, Lipschitz bounds,
discrete counting, stochastic expectations, multi-turn dynamics, and
capacity constraints all point to the same conclusion: under the
wrapper model, some failures persist. 

The practical prescription is to make the boundary shallow, smooth,
and low-dimensional, and to engineer around it rather than assume it
can be eliminated.

\section*{Broader Impact}
This work characterizes structural limitations of a specific class
of defenses (continuous utility-preserving wrappers). The results
could inform defense engineering by identifying which design
constraints matter most. They could also be misread as implying
that LLM defense is futile---this is not the case. The theorems
apply only to wrappers satisfying specific mathematical assumptions;
training-time alignment, architectural changes, discontinuous
filtering, ensemble defenses, and human-in-the-loop systems are
not covered. 

\tt We emphasize that the impossibility results should
motivate better defense design and appropriate policy corrections, not abandonment of defense.

\bibliographystyle{plain}

\begin{thebibliography}{99}

\bibitem{alon2023detecting}
G.~Alon and M.~Kamfonas.
\newblock Detecting language model attacks with perplexity.
\newblock \emph{arXiv preprint arXiv:2308.14132}, 2023.

\bibitem{bagnall2019certifying}
A.~Bagnall and G.~Stewart.
\newblock Certifying the true error: Machine learning in {Coq} with
  verified generalization guarantees.
\newblock \emph{Proceedings of the AAAI Conference on Artificial
Intelligence}, 2019.

\bibitem{bai2022constitutional}
Y.~Bai et~al.
\newblock Constitutional {AI}: Harmlessness from {AI} feedback.
\newblock \emph{arXiv preprint arXiv:2212.08073}, 2022.

\bibitem{chao2024jailbreaking}
P.~Chao, A.~Robey, E.~Dobriban, H.~Hassani, G.~J. Pappas, and
E.~Wong.
\newblock Jailbreaking black box large language models in twenty queries.
\newblock \emph{arXiv preprint arXiv:2310.08419}, 2024.

\bibitem{cohen2019certified}
J.~Cohen, E.~Rosenfeld, and J.~Z. Kolter.
\newblock Certified adversarial robustness via randomized smoothing.
\newblock \emph{Proceedings of ICML}, 2019.

\bibitem{goodfellow2014explaining}
I.~J. Goodfellow, J.~Shlens, and C.~Szegedy.
\newblock Explaining and harnessing adversarial examples.
\newblock \emph{Proceedings of ICLR}, 2015.

\bibitem{inan2023llama}
H.~Inan et~al.
\newblock {Llama Guard}: {LLM}-based input-output safeguard for
  human-{AI} conversations.
\newblock \emph{arXiv preprint arXiv:2312.06674}, 2023.

\bibitem{katz2017reluplex}
G.~Katz, C.~Barrett, D.~L. Dill, K.~Julian, and
M.~J. Kochenderfer.
\newblock Reluplex: An efficient {SMT} solver for verifying deep neural
  networks.
\newblock \emph{Proceedings of CAV}, 2017.

\bibitem{carlini2017evaluating}
N.~Carlini and D.~Wagner.
\newblock Towards evaluating the robustness of neural networks.
\newblock \emph{Proceedings of IEEE S\&P}, 2017.

\bibitem{mehrotra2024tree}
A.~Mehrotra et~al.
\newblock Tree of attacks: Jailbreaking black-box {LLMs} automatically.
\newblock \emph{Advances in Neural Information Processing Systems}, 37,
2024.

\bibitem{fawzi2018adversarial}
A.~Fawzi, H.~Fawzi, and O.~Fawzi.
\newblock Adversarial vulnerability for any classifier.
\newblock \emph{Advances in Neural Information Processing Systems}, 31, 2018.

\bibitem{greshake2023indirect}
K.~Greshake, S.~Abdelnabi, S.~Mishra, C.~Endres, T.~Holz, and M.~Fritz.
\newblock Not what you've signed up for: Compromising real-world
  {LLM}-integrated applications with indirect prompt injection.
\newblock \emph{Proceedings of AISec}, 2023.

\bibitem{huang2017safety}
X.~Huang, M.~Kwiatkowska, S.~Wang, and M.~Wu.
\newblock Safety verification of deep neural networks.
\newblock \emph{Proceedings of CAV}, 2017.

\bibitem{madry2018towards}
A.~Madry, A.~Makelov, L.~Schmidt, D.~Tsipras, and A.~Vladu.
\newblock Towards deep learning models resistant to adversarial attacks.
\newblock \emph{Proceedings of ICLR}, 2018.

\bibitem{mouret2015illuminating}
J.-B. Mouret and J.~Clune.
\newblock Illuminating search spaces by mapping elites.
\newblock \emph{arXiv preprint arXiv:1504.04909}, 2015.

\bibitem{naitzat2020topology}
G.~Naitzat, A.~Zhitnikov, and L.-H.~Lim.
\newblock Topology of deep neural networks.
\newblock \emph{Journal of Machine Learning Research}, 21(184):1--40, 2020.

\bibitem{munshi2026manifold}
S.~Munshi, M.~Bhatt, V.~S.~Narajala, I.~Habler, A.~Al-Kahfah,
K.~Huang, and B.~Gatto.
\newblock Manifold of failure: Behavioral attraction basins in language
  models.
\newblock \emph{arXiv preprint arXiv:2602.22291v2}, 2026.

\bibitem{samvelyan2024rainbow}
M.~Samvelyan et~al.
\newblock Rainbow teaming: Open-ended generation of diverse adversarial
  prompts.
\newblock \emph{Advances in Neural Information Processing Systems}, 37,
2024.

\bibitem{singh2019abstract}
G.~Singh, T.~Gehr, M.~P{\"u}schel, and M.~Vechev.
\newblock An abstract domain for certifying neural networks.
\newblock \emph{Proceedings of POPL}, 2019.

\bibitem{szegedy2013intriguing}
C.~Szegedy et~al.
\newblock Intriguing properties of neural networks.
\newblock \emph{Proceedings of ICLR}, 2014.

\bibitem{tsipras2018robustness}
D.~Tsipras, S.~Santurkar, L.~Engstrom, A.~Turner, and A.~Madry.
\newblock Robustness may be at odds with accuracy.
\newblock \emph{Proceedings of ICLR}, 2019.

\bibitem{wolpert1997no}
D.~H. Wolpert and W.~G. Macready.
\newblock No free lunch theorems for optimization.
\newblock \emph{IEEE Trans.\ Evol.\ Comput.}, 1(1):67--82, 1997.

\bibitem{zou2023universal}
A.~Zou, Z.~Wang, N.~Carlini, M.~Nasr, J.~Z. Kolter, and
M.~Fredrikson.
\newblock Universal and transferable adversarial attacks on aligned
  language models.
\newblock \emph{arXiv preprint arXiv:2307.15043}, 2023.

\bibitem{ge2024mart}
K.~Ge et~al.
\newblock {MART}: Improving {LLM} safety with multi-round automatic
  red-teaming.
\newblock \emph{Proceedings of NAACL}, 2024.

\bibitem{hubinger2024sleeper}
E.~Hubinger et~al.
\newblock Sleeper agents: Training deceptive {LLMs} that persist through
  safety training.
\newblock \emph{arXiv preprint arXiv:2401.05566}, 2024.

\bibitem{anil2024many}
C.~Anil et~al.
\newblock Many-shot jailbreaking.
\newblock \emph{Advances in Neural Information Processing Systems}, 37, 2024.

\bibitem{kim2025manyshot}
D.~Kim et~al.
\newblock What really matters in many-shot attacks?
\newblock \emph{Proceedings of ACL}, 2025.

\bibitem{zhan2024injecagent}
Q.~Zhan et~al.
\newblock {InjecAgent}: Benchmarking indirect prompt injections in
  tool-integrated {LLM} agents.
\newblock \emph{Findings of ACL}, 2024.

\bibitem{zhang2024asb}
H.~Zhang et~al.
\newblock Agent security bench ({ASB}): Formalizing and benchmarking
  attacks and defenses in {LLM}-based agents.
\newblock \emph{arXiv preprint arXiv:2410.02644}, 2024.

\bibitem{yuan2025instability}
Y.~Yuan et~al.
\newblock The instability of safety.
\newblock \emph{arXiv preprint arXiv:2512.12066}, 2025.

\bibitem{skoltech2025quant}
V.~Tsvetkov et~al.
\newblock Quantization and safety: A closer look at {LLM} safety
  under weight compression.
\newblock \emph{arXiv preprint arXiv:2502.15799}, 2025.

\bibitem{eth2025gguf}
F.~Nesti et~al.
\newblock Mind the gap: Adversarial attacks against {GGUF} quantized
  {LLMs}.
\newblock \emph{Proceedings of ICML}, 2025.

\bibitem{hammoud2024merging}
H.~Hammoud et~al.
\newblock Model merging and safety alignment: One bad model spoils the
  bunch.
\newblock \emph{Findings of EMNLP}, 2024.

\bibitem{liu2026whackamole}
X.~Liu et~al.
\newblock Alignment whack-a-mole.
\newblock \emph{arXiv preprint arXiv:2603.20957}, 2026.

\bibitem{iris2025}
D.~Rosenberg et~al.
\newblock {IRIS}: Adversarial suffix attacks against robust defenses.
\newblock \emph{Proceedings of NAACL}, 2025.

\bibitem{zhao2025weak}
X.~Zhao et~al.
\newblock Weak-to-strong jailbreaking on large language models.
\newblock \emph{Proceedings of ICML}, 2025.

\bibitem{huang2025safetytax}
Y.~Huang et~al.
\newblock The safety tax of reasoning alignment.
\newblock \emph{arXiv preprint arXiv:2503.00555}, 2025.

\bibitem{huang2026formal}
Y.~Huang et~al.
\newblock On the geometric inevitability of the alignment tax.
\newblock \emph{arXiv preprint arXiv:2603.00047}, 2026.

\bibitem{slingshot2026}
L.~Bailey et~al.
\newblock Slingshot: {RL}-based agent-to-agent jailbreaking.
\newblock \emph{arXiv preprint arXiv:2602.02395}, 2026.

\end{thebibliography}

\newpage
\appendix
\section*{Appendices}

\section{Vulnerability Landscape}
\label{app:landscape}

This section characterizes the geometry of the unsafe region.

\begin{theorem}[Basin Structure]
\label{thm:basin}
If $f$ is continuous and $f(p)>\tau$, then $U_\tau$ is open. Under
any measure positive on nonempty open sets, $U_\tau$ has positive
measure.
\end{theorem}

\begin{theorem}[Basin Fragment Minimum Size]
\label{thm:fragment}
If $X$ is a normed space, $f$ is $L$-Lipschitz with $f(p) > \tau$,
the connected component of $U_\tau$ containing $p$ has diameter
$\geq 2(f(p) - \tau)/L$.
\end{theorem}

Smoother surfaces (smaller~$L$) produce larger basins; rougher
surfaces produce smaller fragments.

\section{Full Proofs}
\label{app:proofs}

\begin{proof}[Proof of \Cref{thm:main} (Boundary Fixation)]
\textbf{Step~1} (Hausdorff $\Rightarrow$ fixed-point set is closed).
$\Fix(D) = \{x : D(x) = x\}$ is the preimage of the diagonal
$\Delta \subset X \times X$ under $x \mapsto (D(x), x)$. In a
Hausdorff space, $\Delta$ is closed, so $\Fix(D)$ is closed.

\textbf{Step~2} (Utility preservation $\Rightarrow$ safe region
$\subseteq$ fixed points).
$S_\tau \subseteq \Fix(D)$. Since $\Fix(D)$ is closed:
$\cl{S_\tau} \subseteq \Fix(D)$.

\textbf{Step~3} (Connectedness $\Rightarrow$ safe region is not closed).
$S_\tau = f^{-1}((-\infty,\tau))$ is open. If also closed, it would be
clopen---but in a connected space the only clopen sets are $\emptyset$
and $X$. Since both $S_\tau$ and $U_\tau$ are nonempty, $S_\tau$ is not
closed.

\textbf{Step~4} (Boundary point exists).
$\cl{S_\tau} \supsetneq S_\tau$, so there exists
$z \in \cl{S_\tau} \setminus S_\tau$. Continuity gives
$f(z) \leq \tau$; $z \notin S_\tau$ gives $f(z) \geq \tau$. Hence
$f(z) = \tau$.

\textbf{Step~5} (Defense fixes the boundary point).
$z \in \cl{S_\tau} \subseteq \Fix(D)$, so $D(z) = z$ and
$f(D(z)) = \tau$.
\end{proof}

\begin{proof}[Proof of \Cref{thm:eps-robust} ($\varepsilon$-Robust Constraint)]
By \Cref{thm:main}, the fixed boundary point $z$ exists.

\textbf{Step~1.} Since $D(z) = z$ and $D$ is $K$-Lipschitz:
$\dist(D(x), z) = \dist(D(x), D(z)) \leq K\dist(x, z)$.

\textbf{Step~2.} Since $f(z) = \tau$ and $f$ is $L$-Lipschitz:
$|f(D(x)) - \tau| = |f(D(x)) - f(z)| \leq L\dist(D(x), z) \leq
LK\dist(x, z)$.
\end{proof}

\begin{proof}[Proof of \Cref{thm:persistent} (Persistent Unsafe Region)]
\textbf{(1)} $\mathcal{S}$ is the strict superlevel set of
$x \mapsto f(x) - \ell(K+1)\dist(x,z)$ at level~$\tau$, hence open.

\textbf{(2)} Open and nonempty implies positive measure.

\textbf{(3)} For $x \in \mathcal{S}$:
$f(x) > \tau + \ell(K+1)\dist(x,z)$. By \Cref{lem:input-bound}:
$f(D(x)) \geq f(x) - \ell(K+1)\dist(x,z) > \tau$.

\noindent\emph{Lean:
\textup{\texttt{persistent\_unsafe\_refined}} in
\textup{\texttt{MoF\_20\_RefinedPersistence}}
(defense-path constant~$\ell$).
\textup{\texttt{MoF\_11\_EpsilonRobust}} contains an earlier version
using the global constant~$L$; as noted in
\Cref{sec:persistent}, that version is vacuous for isotropic surfaces.}
\end{proof}

\section{Counterexamples: Each Hypothesis Is Necessary}
\label{app:counterexamples}

\begin{counterexample}[Removing connectedness]
$X = \{0,1\}$ discrete, $f(0)=0$, $f(1)=1$, $\tau=0.5$.
$D(0)=0$, $D(1)=0$: continuous, utility-preserving, complete.
\end{counterexample}

\begin{counterexample}[Removing continuity]
$X=[0,1]$, $f(x)=x$, $\tau=0.5$. $D(x)=x$ for $x<0.5$,
$D(x)=0$ for $x \geq 0.5$: utility-preserving and complete,
but discontinuous at $0.5$.
\end{counterexample}

\begin{counterexample}[Removing utility preservation]
$D(x) = x_0$ for a fixed safe point: continuous and complete,
but destroys all inputs.
\end{counterexample}

\section{Attack Properties}
\label{app:attacks}

\begin{theorem}[Perturbation Robustness]
\label{thm:lipschitz}
If $f$ is $L$-Lipschitz and $f(p)>\tau$, then
$B(p,\, (f(p)-\tau)/L) \subseteq U_\tau$.
The radius is monotone in~$f(p)$.
\end{theorem}

\begin{theorem}[Iterative Convergence]
\label{thm:convergence}
Any monotone-improvement operator $T$ with score in $[0,1]$ converges.
If each step gains $\geq \delta$, convergence takes
$\leq \lfloor 1/\delta \rfloor$ steps.
\end{theorem}

\begin{theorem}[Transferability]
\label{thm:transfer}
$\|f-g\|_\infty \leq \delta \implies
\{f > \tau+\delta\} \subseteq \{g > \tau\}$.
Transfer costs zero additional queries.
\end{theorem}

\begin{theorem}[Authority Monotonicity]
\label{thm:authority}
If $f(y,\cdot)$ is monotone non-decreasing in authority for each fixed
indirection~$y$, the vulnerability set is upward-closed with a critical
threshold $a^*_2(y)$ (by IVT, when $f(y,\cdot)$ is continuous and
crosses~$\tau$). If additionally $f$ is monotone non-decreasing in
indirection for each fixed authority, the threshold curve
$y \mapsto a^*_2(y)$ is non-increasing.
\end{theorem}

\begin{theorem}[Gradient Ascent]
\label{thm:gradient}
If $f$ has nonzero Fr\'{e}chet derivative at~$x$, there exist $v$ and
$\varepsilon > 0$ with $f(x + \varepsilon v) > f(x)$.
\end{theorem}

\section{Stability Under Fine-Tuning}
\label{app:stability}

\begin{theorem}[Interior Stability]
\label{thm:interior-stable}
If $\|f - g\|_\infty \leq \varepsilon$:
$f(x) > \tau + \varepsilon \implies g(x) > \tau$, and
$f(x) < \tau - \varepsilon \implies g(x) < \tau$.
Only the band $|f(x) - \tau| \leq \varepsilon$ is uncertain.
\end{theorem}

\begin{theorem}[Crossing Preservation]
\label{thm:crossing}
Let $f, g\colon [a,b] \to \R$ be continuous. If $f$ crosses $\tau$
on $[a,b]$ with margin $m$ (i.e., $f(a') < \tau - m$ and
$f(b') > \tau + m$ for some $a', b' \in [a,b]$) and
$\|f - g\|_\infty \leq \varepsilon < m$, then $g$ also crosses
$\tau$ on $[a,b]$.
\end{theorem}

\begin{theorem}[Patching Is Nonlocal]
\label{thm:nonlocal}
There exist $f, g$ with $\|f-g\|_\infty \leq \varepsilon$ such that
eliminating a vulnerability at one point necessarily changes values at
distant points.
\end{theorem}

\section{Cost Asymmetry}
\label{app:cost}

\begin{theorem}[Exponential Cost Asymmetry]
\label{thm:cost}
For grid-based defense with $N \geq 2$ bins per axis in $d$ dimensions:
attack cost $\leq 1/\delta$ (dimension-independent); defense cost
$= N^d$ (exponential in~$d$); ratio
$\delta \cdot N^d \to \infty$ as $d\to\infty$.
\end{theorem}

At $d=2$ with $N=25$ and $\delta=0.01$, the ratio is $6.25$.
At $d=10$, it climbs to $\sim 10^{12}$.

\section{Additional Verified Results}
\label{app:additional}

\paragraph{Lipschitz displacement bound.}
$D$ is $K$-Lipschitz with $D(z) = z \implies
\dist(D(x), x) \leq (K+1)\dist(x,z)$.

\paragraph{Tool calls amplify failure.}
Each non-contractive tool call multiplicatively increases the pipeline's
effective Lipschitz constant: $K^n < K^{n+1}$ for $K \geq 2$.

\paragraph{Attacker monotone improvement.}
Best observed alignment deviation is non-decreasing across turns.

\paragraph{Attacker steering.}
If directional slope varies continuously with attacker parameter $\alpha$
and crosses $\ell(K+1)$, transversality is reachable by IVT.

\paragraph{Stochastic regularity.}
$g(x) = \mathbb{E}[f(D(x))]$ satisfies $g - f = 0$ on $\cl{S_\tau}$.

\paragraph{Discrete defense dilemma.}
An injective, utility-preserving defense is incomplete (every unsafe
input stays unsafe). A complete, utility-preserving defense is
non-injective (distinct inputs collapse to the same output).
The three properties---completeness, utility preservation,
injectivity---form a trilemma.

\paragraph{Defense position invariance.}
Lipschitz constant of defense-before-tools equals defense-after-tools:
$K_D \cdot K_T^n$.

\section{Lean Artifact}
\label{app:artifact}

The complete theory is verified in Lean~4.28.0 with Mathlib~v4.28.0,
available at
\url{https://github.com/mbhatt1/stuff/tree/main/ManifoldProofs}.
The artifact comprises 46 files:
\begin{itemize}[nosep]
  \item 10 core theory files (MoF\_01--MoF\_10)
  \item 10 cost theory files (MoF\_Cost\_01--MoF\_Cost\_10)
  \item 10 advanced theory files (MoF\_Adv\_01--MoF\_Adv\_10)
  \item 1 continuous relaxation (MoF\_ContinuousRelaxation)
  \item 1 $\varepsilon$-robust constraint + persistent unsafe region
    with global Lipschitz constant
    (MoF\_11\_EpsilonRobust)
  \item 1 discrete impossibility (MoF\_12\_Discrete)
  \item 1 multi-turn + stochastic extensions
    (MoF\_13\_MultiTurn)
  \item 1 representation-independent meta-theorem
    (MoF\_14\_MetaTheorem)
  \item 1 nonlinear agent pipelines (MoF\_15\_NonlinearAgents)
  \item 1 relaxed utility preservation (MoF\_16\_RelaxedUtility)
  \item 1 quantitative $\varepsilon$-band volume bound
    (MoF\_17\_CoareaBound)
  \item 1 cone measure bound for persistent unsafe region
    (MoF\_18\_ConeBound)
  \item 1 optimal defense characterization
    (MoF\_19\_OptimalDefense)
  \item 1 refined persistence with defense-path Lipschitz constant~$\ell$
    (MoF\_20\_RefinedPersistence); this file contains the primary
    formalization of \Cref{thm:persistent,lem:input-bound,def:steep}
  \item 1 gradient chain: $\|\nabla f(z)\| > \ell(K+1)$ implies
    persistent unsafe region via operator-norm direction extraction
    and derivative-based local growth
    (MoF\_21\_GradientChain)
  \item 3 capstone files (MasterTheorem, Euclidean instantiation,
    verification)
  \item 1 root import file (ManifoldProofs.lean)
\end{itemize}

\end{document}